\documentclass[cleveref,a4paper,USenglish,thm-restate]{lipics-v2021}
\hideLIPIcs
\usepackage[utf8]{inputenc}
\usepackage{cite}
\usepackage{todonotes,xspace}
\usepackage[outercaption]{sidecap}
\usepackage{hyperref}

\renewcommand{\phi}{\varphi}
\renewcommand{\epsilon}{\varepsilon}

\newcommand{\Ohs}{\mathcal{O}^*}
\newcommand{\cS}{\mathcal{S}}
\newcommand{\cM}{\mathcal{M}}

\newcommand{\bigoh}{\mathcal{O}}
\newcommand{\bigohs}{\mathcal{O^*}}

\renewcommand{\leq}{\leqslant}
\renewcommand{\geq}{\geqslant}
\renewcommand{\le}{\leqslant}

\newcommand{\homo}[1]{\textsc{Hom}(\ensuremath{#1})\xspace}
\newcommand{\lhomo}[1]{\textsc{LHom}(\ensuremath{#1})\xspace}

\newcommand{\hcoloringext}[1]{\textsc{HomExt}(\ensuremath{#1})\xspace}
\newcommand{\qcsp}[2]{\ensuremath{#1}-\textsc{CSP}-\ensuremath{#2}\xspace}

\newcommand{\cw}[1]{{\operatorname{cw}(#1)}}

\newcommand{\tws}{{\operatorname{tw}}}

\newcommand{\cws}{{\operatorname{cw}}}

\newtheorem{myconjecture}{Conjecture}

\nolinenumbers

\title{The Fine-Grained Complexity of Graph Homomorphism Parameterized by Clique-Width}
\titlerunning{The Fine-Grained Complexity of Graph Homomorphism Parameterized by Clique-Width}

\author{Robert Ganian}{Algorithms and Complexity Group, TU Wien, Vienna, Austria}{rganian@ac.tuwien.ac.at}{0000-0002-7762-8045}{Robert Ganian acknowledges support by the Austrian Science Fund (FWF, projects Y1329 and P31336).}
\author{Thekla Hamm}{Algorithms and Complexity Group, TU Wien, Vienna, Austria}{thamm@ac.tuwien.ac.at}{0000-0002-4595-9982}{Thekla Hamm acknowledges support by the Austrian Science Fund (FWF, projects P31336 and Y1329).}
\author{Viktoriia Korchemna}{Algorithms and Complexity Group, TU Wien, Vienna, Austria}{vkorchemna@ac.tuwien.ac.at}{}{Viktoriia Korchemna acknowledges support by the Austrian Science Fund (FWF, project Y1329).}
\author{Karolina Okrasa}{Faculty of Matematics and Information Science, Warsaw University of Technology, Warsaw, Poland, and Faculty of Mathematics, Informatics and Mechanics, University of Warsaw, Poland}{k.okrasa@mini.pw.edu.pl}{0000-0003-1414-3507}{Karolina Okrasa acknowledges support by the European Research Council, grant agreement No 714704. Parts of this work were performed while visiting TU Wien, Vienna, Austria.}
\author{Kirill Simonov}{Algorithms and Complexity Group, TU Wien, Vienna, Austria}{ksimonov@ac.tuwien.ac.at}{}{Kirill Simonov acknowledges support by the Austrian Science Fund (FWF, project P31336).}

\authorrunning{R.~Ganian, T.~Hamm, V.~Korchemna, K.~Okrasa, K.~ Simonov}

\keywords{homomorphism, clique-width, fine-grained complexity}
\Copyright{Robert Ganian, Thekla Hamm, Viktoriia Korchemna, Karolina Okrasa, Kirill Simonov}

\ccsdesc[300]{Theory of computation~Parameterized complexity and exact algorithms}

\begin{document}
\maketitle

\begin{abstract}
The generic homomorphism problem, which asks whether an input graph $G$ admits a homomorphism into a fixed target graph $H$, has been widely studied in the literature. In this article, we provide a fine-grained complexity classification of the running time of the homomorphism problem with respect to the clique-width of $G$ (denoted $\cws$) for virtually all choices of $H$ under the Strong Exponential Time Hypothesis. In particular, we identify a property of $H$ called the signature number $s(H)$ and show that for each $H$, the homomorphism problem can be solved in time $\bigohs(s(H)^{\cws})$. Crucially, we then show that this algorithm can be used to obtain essentially tight upper bounds. Specifically, we provide a reduction that yields matching lower bounds for each $H$ that is either a projective core or a graph admitting a factorization with additional properties---allowing us to cover all possible target graphs under long-standing conjectures.
\end{abstract}
\section{Introduction}
A \emph{homomorphism} from a graph $G$ to a graph $H$ is an edge-preserving mapping from the vertices of $G$ to the vertices of $H$. Homomorphisms are fundamental constructs which have been studied from a wide variety of perspectives~\cite{Grohe07,BulatovD20,Boker21}.
Our focus here will be on the class of problems which ask whether an input $n$-vertex graph $G$ admits a homomorphism to a fixed target graph $H$. This ``meta-problem''---which we simply call \homo{H}---captures, among others, the classical $c$-\textsc{Coloring} problems when $H$ is set to the complete graph on $c$ vertices. Famously, Hell and Ne\v{s}et\v{r}il~\cite{DBLP:journals/jct/HellN90} proved that \homo{H} is polynomial-time solvable if $H$ is bipartite or has a loop, and \textsf{NP}-complete otherwise.
While the aforementioned result provides a basic classification of the complexity of \homo{H}, it does not say much in terms of how quickly one can actually solve these problems. Indeed, the usual assumption that \textsf{P} $\neq$ \textsf{NP} is not sufficient to obtain tight bounds for the running times of algorithms. While upper bounds can be straightforwardly obtained by designing a suitable algorithm, the corresponding lower bounds usually rely on the Exponential Time Hypothesis (ETH) or the Strong Exponential Time Hypothesis (SETH), which allows for even tighter bounds~\cite{DBLP:journals/jcss/ImpagliazzoPZ01, ImpagliazzoP01, LokshtanovMS11}.
It is not difficult to design a brute-force algorithm for the homomorphism problem that runs in time $\bigohs(|V(H)|^n)$ for every choice of $H$, and thanks to the breakthrough result of Cygan et al. we now know that this running time is essentially tight under the 
Exponential Time Hypothesis (ETH)~\cite{DBLP:journals/jacm/CyganFGKMPS17} as long as one considers only the dependency on $n$ and $|V(H)|$.
Still, it is often possible to circumvent this lower bound and obtain significantly better runtime guarantees. One approach to do so is to consider restrictions on the class of targets: if $H$ is a complete graph then \homo{H} can be solved in time $\bigohs(2^n)$, and there are also several algorithms that achieve running times of the form $\bigohs(\alpha(H)^n)$ where $\alpha(H)$ is some structural parameter of $H$~\cite{DBLP:journals/mst/FominHK07, DBLP:journals/mst/Wahlstrom11, Rzazewski14}.
The other is to exploit the properties of the input graph $G$, which are commonly captured by a suitably defined structural parameter. The most commonly used graph parameter in this respect is \emph{treewidth}~\cite{RobertsonS86}, which informally measures how ``tree-like'' a graph is.
When considering treewidth, it is once again not difficult to obtain an algorithm that runs in time $\bigohs(|V(H)|^\tws)$, where $\tws$ is the treewidth of $G$; as before, it was much more difficult to show that this is essentially optimal. The first SETH-based tight lower bound in this setting was actually shown for special cases of the related problem of \lhomo{H}, where each vertex in the graph $G$ comes with a list of admissible targets for the homomorphism~\cite{DBLP:conf/stacs/EgriMR18}; this was later lifted to a full classification~\cite{OkrasaPR20}.
A nearly-complete SETH-based lower bound result for \homo{H} itself was only obtained recently by Okrasa and Rz\k{a}\.zewski~\cite{OkrasaSODAJ}; in particular, the result covers all targets which are so-called \emph{projective cores}. It is known that almost all graphs are projective cores~\cite{HellN92,LuczakN04,OkrasaSODAJ}, 
and it is worth noting that the authors showed that their result can be lifted to all targets under long-standing conjectures on the properties of projective cores~\cite{LaroseT01,Larose02}.
While treewidth is the most prominent structural graph parameter, it is not the most general\footnote{There is a hierarchy of graph parameters~(see, e.g., \cite[Figure 1]{Belmonte0LMO20}), 
where parameter $\mathbb{A}$ is more general than parameter $\mathbb{B}$ if there are graph classes of bounded $\mathbb{A}$ and unbounded $\mathbb{B}$ but the opposite is not true.} one that can be used to efficiently solve \homo{H}. Indeed, standard dynamic programming techniques can be used to obtain a $\bigohs((2^{|V(H)|})^\cws)$ time algorithm for the problem, where $\cws$ stands for \emph{clique-width}~\cite{CourcelleMR00}: a well-studied graph parameter that is bounded not only on all graph classes of bounded treewidth, but also on well-structured dense classes such as complete graphs. But is this basic algorithm generally optimal (mirroring the situation for treewidth~\cite{OkrasaSODAJ}), or can one obtain better runtime dependencies on clique-width?

\smallskip

\noindent \textsf{\bfseries Contribution.} \quad
Our aim is to obtain a detailed understanding of the fine-grained complexity of \homo{H} in terms of  the clique-width of $G$ and the fixed target $H$. As a starting point for our investigation, we note that Lampis used the SETH to obtain tight bounds for \textsc{$c$-Coloring} with respect to clique-width~\cite{DBLP:journals/siamdm/Lampis20}. Interestingly, already for this special case, the upper and lower bounds differ from those of the aforementioned simple dynamic programming algorithm: if $H$ is a complete graph, then \homo{H} can be solved in time $\bigohs((2^{|V(H)|}-2)^\cws)$~\cite{DBLP:journals/siamdm/Lampis20} and this is tight under the SETH. However, as noted by Piecyk and Rzążewski~\cite{PiecykR21}, it was not at all obvious how these bounds can be lifted to general choices of $H$.
In order to achieve our goals we need to improve upon the basic dynamic programming idea to identify a ``hopefully correct'' base of the exponent for every choice of $H$. Towards our first result, we identify a structural property of $H$ called the \emph{signature  number} (denoted $s(H)$) which, intuitively, captures the number of non-trivial neighborhood classes of vertex subsets in $H$ (the \emph{signature set}). We then obtain a non-trivial dynamic programming algorithm that solves \homo{H} in time where the base of the exponent is precisely the signature number. We note that $s(H)$ is $2^{|V(H)|}-2$ for complete graphs $H$, and so this result also provides a succinct and broader explanation for the running time of Lampis' algorithm~\cite{DBLP:journals/siamdm/Lampis20}.
\begin{restatable}{theorem}{thmalgorithm}
\label{thm:algorithm-main}
	Let $H$ be a fixed graph. \homo{H} can be solved in time $\Ohs(s(H)^\cw{G})$ for each input graph $G$, assuming an optimal clique-width expression of $G$ is provided as part of the input.
\end{restatable}
With this upper bound, we proceed to the main technical contribution of this paper:  establishing a corresponding lower bound under the SETH. The main difficulty here is that we need a reduction that is delicate on one hand, since it needs to preserve the clique-width, but is on the other hand also flexible enough to work for many different choices of $H$; moreover, the reduction has to rely on the signature numbers of these graphs in some way. 
To provide an intuitive description of the reduction, let us focus for now on the case where $H$ is a projective core. On a high level, the main building block is an \emph{$S$-gadget} which, given an arbitrary set $S$ of pairs of vertices in $H$ and two vertices $p$ and $q$ of the input graph $G$, ensures that every homomorphism $f$ satisfies $(f(p),f(q))\in S$. After providing a generic construction for such $S$-gadgets which is clique-width preserving and works for every valid choice of $H$, we use these to obtain \emph{implication gadgets} and \emph{or gadgets} which restrict how a solution homomorphism can behave on a selected set of vertices in $G$. The formalization of these gadgets is the main technical hurdle towards the desired result; once that is done, we can lift the idea used in the earlier reduction of Lampis~\cite{DBLP:journals/siamdm/Lampis20} that established clique-width lower bounds for \textsc{$c$-Coloring} by reducing from \textsc{Constraint Satisfaction} (\textsc{CSP}) to \homo{H}. One crucial distinction in our reduction is that we use elements of the signature set (as opposed to color sets) to represent domain values in the \textsc{CSP} instance.
To lift these considerations to cases where $H$ is not a projective core, we unfortunately need to add an extra layer of complexity. Similarly as in the previous treewidth-based lower bound for \homo{H}~\cite{OkrasaSODAJ}, one can base this step on conjectures of Larose and Tardif~\cite{LaroseT01,Larose02} that classify all remaining targets as certain graph products with special properties (notably, all of the factors must be ``truly projective'').
The approach used for treewidth~\cite{OkrasaSODAJ} was then to essentially repeat all steps of the proof for projective cores, with the added difficulty that one uses the properties of products instead of dealing directly with projective cores.
While this approach could be used here as well, instead we unify the two cases ($H$ being a projective core, and $H$ being a product) by defining the notion of $W$-projectivity for some factor $W$ of $H$. In particular, if $H$ is a projective core then it itself is $H$-projective, while if $H$ is a product with truly projective factor $H_i$ then it is $H_i$-projective. As our main result, we obtain an SETH-based lower bound which essentially shows that for each $W$-projective graph $H$, $s(W)$ is the optimal base of the clique-width exponent for solving $\homo{H}$:
\begin{restatable}{theorem}{thmlowerbound}
\label{thm:lower-bound-homo}
If $H$ is $H_i$-projective for some $i\in [m]$ then there is no algorithm solving \homo{H} in time $\Ohs((s(H_i)-\epsilon)^{\cw{G}})$ for any $\epsilon >0$, unless the \textup{SETH} fails.
\end{restatable}

By also deliberately considering prime factorizations in the algorithm which we provide for Theorem~\ref{thm:algorithm-main}, we can obtain an upper bound on the complexity of \homo{H} that matches the lower bound from Theorem~\ref{thm:lower-bound-homo}.
For a discussion explicitly relating these complexity bounds in the context of the aforementioned conjectures of Larose and Tardif, we refer to Section~\ref{sec:conclusion}.

\section{Preliminaries}
We use standard terminology for graph theory~\cite{Diestel12}. Let $[i]$ 
denote the set $\{1,\dots,i\}$. For a mapping $f:A\rightarrow B$ and $A'\subseteq A$, let $f|_{A'}$ denote the restriction of $f$ to $A'$.
We will use the $\Ohs(\cdot)$ notation to suppress factors polynomial in the input size.

\subsection*{Homomorphisms and Cores}
For two graphs $G$ and $H$, a \emph{homomorphism} from $G$ to $H$ is a mapping $h: V(G) \to V(H)$, such that for every $uv \in E(G)$ we have $h(u)h(v) \in E(H)$.
If there exists a homomorphism from $G$ to $H$, we denote this fact by $G \to H$, and if $h$ is a homomorphism from $G$ to $H$, we denote that by $h: G \to H$.
If there is no homomorphism from $G$ to $H$, we write $G \not\to H$.
If $G \to H$ and $H \to G$, we say that $G$ and $H$ are \emph{homomorphically equivalent}.
In particular, since the composition of homomorphisms is a homomorphism, if $G$ and $H$ are homomorphically equivalent, then for every graph $F$ we have that $F \to G$ if and only if $F \to H$. 
It is straightforward to verify that homomorphic equivalence is an equivalence relation on the class of all graphs.
On the other hand, if $G \not\to H$ and $H \not\to G$ for some graphs $G,H$, we say that $G$ and $H$ are \emph{incomparable}.

We note that if $H$ is a clique on $c$ vertices, then homomorphisms form $G$ to $H$ are precisely proper vertex $c$-colorings of $G$. 

We say that a graph $H$ is a \emph{core} if every homomorphism $h: H \to H$ is an automorphism.
Equivalently, $H$ is a core if for every proper induced subgraph $H'$ of $H$ it holds that $H \not\to H'$.
We say that a core $H'$ is a \emph{core of $H$} if $H'$ is an induced subgraph of $H$ and $H \to H'$.
Clearly, each core graph is a core of itself.
Each graph has a unique (up to isomorphism) core, and the core of $H$ can be equivalently defined as the smallest (with respect to the number of vertices) graph that is homomorphically equivalent with $H$~\cite{HellN92}.

A graph $H$ is \emph{ramified} if $N(u) \not\subseteq N(v)$ for every two distinct vertices $u,v$ of $H$.
Observe that each core is ramified; otherwise one could define $f:H \to H$ that is an identity on all vertices of $H$ but $u$ and set $f(u)=v$. This would be a homomorphism to a proper subgraph of $H$, contradicting the fact that $H$ is a core.

We say that a graph $H$ is \emph{trivial} if its core has at most two vertices.
\begin{observation}[\hspace{-0.001cm}\cite{DBLP:journals/jct/HellN90}]
\label{obs:trivial-cores}
A graph $H$ is trivial if and only if it is either bipartite or contains a vertex with a loop.
\end{observation}
\begin{proof}
It is straightforward to observe that there exist three trivial cores: $K_1$, $K_2$, and $K^*_1$, where by $K_1^*$ we denote the graph that consists of one vertex with a loop. 

If $H$ contains a vertex $a$ with a loop, then $K_1^*$ is the core of $H$, as mapping every vertex of $H$ to $v$ yields a homomorphism. 
If $H$ is bipartite, then the core of $H$ is either $K_1$ (if $H$ has no edges) or $K_2$ (since mapping the vertices of one bipartition class to one vertex of $K_2$, and another bipartition class to the other, is a homomorphism).

For the other direction, assume that $H$ is a non-bipartite loopless graph.
Since it is loopless, $K_1^*$ cannot be its core.
Clearly, $H$ has at least one edge, and therefore $H \not\to K_1$
Moreover, $H$ contains an odd cycle $C_{2k+1}$ as a subgraph, hence, $C_{2k+1} \to H$.
If now $H \to K_2$, composition of these homomorphism gives that $C_{2k+1} \to K_2$, which is equivalent to stating that $C_{2k+1}$ is 2-colorable, a contradiction. 
\end{proof}
Observe that trivial cores $H$ correspond precisely to the polynomial cases of the \homo{H} problem. Since our aim is to focus on the \textsf{NP}-hard cases of the problem, from here onward we will assume that the target graph is non-trivial.

\subsection*{Signature Sets}
For a vertex $v$ of a graph $H$, let $N_H(v)$ denote the set of neighbors of $v$ in $H$. If the graph is clear from the context, we will omit the
subscript $H$ and write $N(v)$.

For a non-empty set $T \subseteq V(H)$ we say that $S(T)$ is the \emph{signature set of $T$} if $S(T)=\bigcap_{t \in T} N(t)$.
We say that a non-empty set $S \subseteq V(H)$ is a \emph{signature set}, if there exists $T$ such that $S=S(T)$.
We denote by $\cS(H)$ the set of all signature sets of $H$, and we note that $\{V(H),\emptyset\}\cap \cS(H)=\emptyset$.
\begin{observation}\label{obs:st-mst}
If $T$ is a proper non-empty subset of $V(H)$, and $a \in T$, $b \in S(T)$, then $ab \in E(H)$. Moreover, for non-empty subsets $A,B \subseteq V(H)$, $S(A \cup B)=S(A)\cap S(B)$. 
\end{observation}

We note that the operation of taking a signature set is reversible on $\cS(H)$:

\begin{observation}
\label{obs:sh-reverse}
For every $A \in \mathcal{S}(H)$, $S(S(A))=A$. 
\end{observation}
\begin{proof}
By the definition of signature set, $A\times S(A) \subseteq E(H)$, so $A\subseteq S(S(A))$. For the converse direction observe that as $A \in \mathcal{S}(H)$, there exists a non-empty subset $T$ of $V(H)$ such that $A=S(T)$. Pick any $x\in S(S(A))$, then $ E(H) \supseteq \{x\}\times S(A)=\{x\}\times S(S(T)) \supseteq \{x\}\times T$. Hence by definition $x\in S(T)=A$.
\end{proof}

Let the \emph{signature number} of $H$, denoted $s(H)$, be defined as $|\cS(H)|$. As mentioned in the introduction, the signature number will play a crucial role in our upper and lower bounds.

Observe that, if $H$ is a target and hence non-trivial, for every nonempty $T \subseteq V(H)$ we have that $S(T) \cap T = \emptyset$. 
From that it is easy to see that $V(H)$ never belongs to $\cS(H)$. 
Since, by definition, $\emptyset \notin \cS(H)$, we get the following bounds for $s(H)$.

\begin{observation}\label{obs:sh-bound}
Let $H$ be a graph with no loops.
Then $s(H) \leq 2^{|V(H)|}-2$.
\end{observation}
Notice that since $2^{|V(H)|}-2$ is the number of all proper non-empty subsets of $V(H)$, the equality in \Cref{obs:sh-bound} holds if and only if $H$ is a clique.

If $S \in \cS(H)$, we call $T$ such that $S(T)=S$ a \emph{witness} of $S$.
Clearly, we can have distinct $T_1,T_2$ such that $S(T_1)=S(T_2)$, however, notice that in such a case there exists $T=T_1 \cup T_2$ such that $S(T)=S(T_1)=S(T_2)$.
Hence, there exists a unique maximal (with respect to inclusion) witness of $S$, and we denote it by $M(S)$.
In fact, it is not difficult to see that \(M(S) = \{v \in V(H) \mid S \subseteq N_H(v)\}\);
for \(S(M(S)) = S\) to hold, it is clearly necessary that \(S \subseteq N_H(v)\) for all \(v \in M(S)\).
On the other hand, as \(M(S)\) is maximal all \(v\) for which this is true are contained in \(M(S)\).

In this way signature sets and their witnesses are in one-to-one correspondence.
While not necessary to obtain our algorithmic and lower bounds for \homo{H} parameterized by clique-width, this offers an alternative perspective on the role of signatures in our results.

In fact, the signature number could equivalently be defined as the `maximal witness number' and signature sets could be replaced by maximal witnesses in all our proofs: 
\begin{observation}\label{obs:sh-mh}
	Let $\cM(H)=\{M(S): S \in \cS(H)\}$, then $\cS(H)=\cM(H)$.
\end{observation}
\begin{proof}
	Let $T$ be a fixed non-empty subset of $V(H)$ such that $S(T) \neq \emptyset$.
	Observe that the definition and the maximality of $M(S(T))$ implies that $S(S(T))=\bigcap_{t \in S(T)} N(t)=M(S(T))$. 
	Since $T$ and $S(T)$ are non-empty, we get that $\cS(H)\supseteq\cM(H)$. For the other direction observe that $M(M(S(T))) = S(T)$. Assume the contrary, then there exists a vertex $v \notin S(T)$ such that $N(v) \supset M(S(T))$. However, $T \subseteq M(S(T)) \subseteq N(v)$ meaning that $v \in S(T)$, a contradiction. Thus, $\cS(H) = \cM(H)$.
\end{proof}

We note that if $H$ is a core graph, we can also bound the minimum cardinality of $\cS(H)$. 

\begin{observation}

\label{obs:sh-cores}
Let $H$ be a core graph, $H \neq K_1$. 
Then $s(H) \geq |V(H)|$.
\end{observation}
\begin{proof}
Observe that if $H$ is a core distinct from $K_1$, then it does not contain isolated vertices.
Therefore, for each $v \in V(H)$ we have $N(v) \in \mathcal{S}(H)$.
On the other hand, since $H$ is a core, it is ramified.
In particular, for every distinct $v,w \in V(H)$ we have $N(v)\ne N(w)$. Hence different vertices give rise to different signature sets. 
\end{proof}

\subsection*{Clique-Width and Clique-Width Expressions}
For a positive integer $k$, we let a \emph{$k$-graph} be a graph whose vertices
are labeled by $[k]$. For convenience, we consider a graph to be a $k$-graph with all vertices labeled by $1$. 
We call the $k$-graph consisting of exactly one vertex $v$ (say,
labeled by $i$) an initial $k$-graph and denote it by $i(v)$.

The \emph{clique-width} of a graph $G$ is the smallest integer $k$ such that
$G$ can be constructed from initial $k$-graphs by means of iterative
application of the following three operations:
\begin{enumerate}
	\item Disjoint union (denoted by $\oplus$);
	\item Relabeling: changing all labels $i$ to $j$ (denoted by $\rho_{i\rightarrow j}$);
	\item Edge insertion: adding an edge from each vertex labeled by $i$
	to each vertex labeled by $j$ ($i\neq j$; denoted by
	$\eta_{i,j}$).
\end{enumerate}
A construction of a $k$-graph $G$ using the above operations can be
represented by an algebraic term composed of $\oplus$,
$p_{i\rightarrow j}$ and $\eta_{i,j}$ (where $i\neq j$ and $i,j\in
[k]$). Such a term is called a $k$-expression defining $G$, and we often view it as a tree with each node labeled with the appropriate operation.
Conversely, we call the $k$-graph that arises from a \(k\)-expression its \emph{evaluation}.
The \emph{clique-width} of $G$ is the smallest integer $k$ such that $G$
can be defined by a $k$-expression which we then also call a \emph{clique-width expression} of \(G\).

Many graph classes are known to have  constant clique-width; examples
include all graph classes of constant treewidth and
co-graphs~\cite{CourcelleO00}.
Moreover a fixed-parameter algorithm is known to compute a \(k\)-expression of the input where \(k\) is bounded in \(f(\cws)\)~\cite{Oum05}.

\section{Algorithm}
As our first contribution, we obtain an algorithm that will play a crucial role for upper-bounding the fine-grained complexity of \homo{H}.

\thmalgorithm*
\begin{proof}
	Assume, w.l.o.g., that $G$ is connected and $|V(G)|>1$. We will describe a dynamic program that proceeds in a leaf-to-root fashion along the provided $k$-expression \(\sigma\) of $G$.
	For a subexpression \(\tau \subseteq \sigma\), we denote the evaluation of \(\tau\) by \(G_\tau\), and by \(V^i_\tau \subseteq V(G_{\tau})\) the vertex set that has label \(i\) in \(G_\tau\).
	We say that \(i\) is a \emph{live label} in \(\tau\) if there is an edge of \(G\) which is incident to \(V^i_\tau\) and does not appear in \(G_\tau\).
	Denote the set of live labels in \(\tau\) by \(L_\tau\). Since $G$ is connected, \(L_\tau\ne \emptyset\) for any proper subexpression $\tau$ of $\sigma$.
	
	For each subexpression \(\tau\) of \(\sigma\), we will compute a set $P_\tau$ consisting of functions $p \colon L_\tau \to \mathcal{S}(H)$ where $p\in P_\tau$ if and only if there exists a homomorphism $h_p$ from $G_\tau$ to $H$ such that \(p(i) \subseteq S(h_p(V^i_\tau))\), \(i \in L_\tau\). We will say that $p \in P_\tau$ \emph{describes} the homomorphism $h_p$ in $\tau$ or, equivalently, that $h_p$ \emph{witnesses} $p$ in $\tau$. 
		Intuitively, we will use $p(i)$ to preemptively store the images of the neighbors of $V^i_\tau$ in the final graph $G$---that is why we store not only the exact signature, but all signatures that occur as subsets. We remark that storing the ``current'' images of the neighbors of $V^i_\tau$ in $G_\tau$ would be sufficient to obtain a conceptually simpler fixed-parameter algorithm parameterized by clique-width, but in that case it is not obvious how one can avoid a quadratic dependency on clique-width in the exponent.
		
Observe that for any homomorhism $h \colon G_\tau \to H$, images of vertices with live labels should be connected in $H$ with images of their future neighbors. In particular, for any \(i \in L_\tau\), $S(h(V^i_\tau))\ne \emptyset$ and hence $S(h(V^i_\tau))\in \mathcal{S}(H)$. Therefore $h$ is described in $\tau$ by some $p\in P_{\tau}$. 
By definition, $L_\sigma=\emptyset$ and hence $G$ is homomorphic to \(H\) if and only if $P_\sigma$ contains the empty mapping, i.e., if $P_\sigma=\{ \emptyset \}$ (as opposed to $P_\sigma = \emptyset$). It remains to show how to correctly compute each \(P_\tau\).
	To do so, we distinguish based on the outermost operation of \(\tau\):
	
	\paragraph*{\(\tau = i(v)\) for some \(i \in [\cw{G}]\).}
	In this case \(L_\tau = \{i\}\), and \(P_\tau\) contains all functions $p: \{i\} \mapsto \mathcal{S}(H)$  such that $p(i) \subseteq N_H(u)$ for some $u\in V(H)$. 
	
	\paragraph*{\(\tau = \rho_{i \to j}(\tau')\) and \(P_{\tau'}\) has already been computed.}
	If \(i \not \in L_{\tau'}\),
	we can correctly set
	\(L_\tau = L_{\tau'}\) and $P_\tau=P_{\tau'}$. If $i\in L_{\tau'}$ and \(j \not \in L_{\tau'}\), then \(L_\tau = (L_{\tau'} \setminus \{i\}) \cup \{j\}\) and \[
	P_\tau = \left\{p \colon L_\tau \to \mathcal{S}(H) ~\mid~ \exists p' \in P_{\tau'}:
     \quad p(\ell) = 
    \begin{cases}
	p'(\ell) & \mbox{if } \ell\neq j \\ 
	p'(i) & \mbox{if } \ell=j
	\end{cases}\right\}.
	\] 
	Finally, if $\{i,j\}\subseteq L_{\tau'}$, then \(L_\tau = L_{\tau'} \setminus \{i\} \) and $
	P_\tau = \{p'|_{L_\tau} ~\mid~  p' \in P_{\tau'} \land p'(i)=p'(j)\}$. 
		
	For correctness in the last case, let \(h\) be a homomorphism from $ G_{\tau}$ to \(H\) and $S_{\ell}\in \mathcal{S}(H)$ be such that $S_{\ell}\subseteq S(h(V^{\ell}_{\tau}))$, $\ell \in L_{\tau}$. Observe that $V^{\ell}_{\tau}=V^{\ell}_{\tau'}$ for $\ell \in L_{\tau}\setminus \{j\}$ and $V^{j}_{\tau}=V^{j}_{\tau'}\cup V^{i}_{\tau'}$. In particular, $S_j\subseteq S(h(V^{i}_{\tau'}))$ and $S_j\subseteq S(h(V^{j}_{\tau'}))$. By definition of $P_{\tau'}$, there exists \(p' \in P_{\tau'}\) such that $p'(\ell) = S_{\ell}$ for $\ell \in L_{\tau}\setminus \{j\}$ and \(p'(i)=p'(j) = S_j\).
	The function \(p \in P_\tau\), defined by \(p=p'|_{L_\tau}\), satisfies $p(\ell) = S_{\ell}$ for each $\ell \in L_{\tau}$.
			
	On the other hand, fix some \(p \in P_\tau\).
	Let \(p' \in P_{\tau'}\) be a function such that \(p\) arises from \(p'\) in the construction of \(P_\tau\).
	Consider a witness $h$ of $p'$ in $\tau'$. For every $\ell \in L_{\tau}\setminus \{j\}$ we have $V^{\ell}_{\tau}= V^{\ell}_{\tau'}$ and so $p(\ell)=p'(\ell)\subseteq S(h(V^{\ell}_{\tau}))$. Moreover, $p(j)=p'(j)\subseteq S(h(V^{j}_{\tau'}))$ and $p(j)=p'(i)\subseteq S(h(V^{i}_{\tau'}))$. By \Cref{obs:st-mst}, we have $p(j)\subseteq S(h(V^{i}_{\tau'}))\cap S(h(V^{j}_{\tau'})) = S(h(V^{i}_{\tau'}\cup V^{j}_{\tau'}))= S(h(V^{j}_{\tau}))$.  Hence $p$ witnesses $h$ in $\tau$.
	 
	\paragraph*{\(\tau = \tau^{(1)} \oplus \tau^{(2)}\) where \(P_{\tau^{(1)}}\) and \(P_{\tau^{(2)}}\) have already been computed.}
	In this case \(L_\tau = L_{\tau^{(1)}} \cup L_{\tau^{(2)}}\) and
		$$P_\tau = \{p =p_1\cup p_2 ~\mid~ p_1 \in P_{\tau^{(1)}} \wedge p_2 \in P_{\tau^{(2)}} \wedge
	        (\forall \ell \in L_{\tau^{(1)}}\cap L_{\tau^{(2)}} \colon p_1(\ell)=p_2(\ell))\}$$
Intuitively, we construct a homomorphism on the disjoint union of two subgraphs by ``gluing together'' the homomorphisms on the subgraphs. If the subgraphs share any live labels, after this step they will all be treated equally. For this reason we require the images of the neighbors of such labels to be the same in both subgraphs.
For correctness, let \(h\) be a homomorphism from $ G_{\tau}$ to \(H\) and $S_{\ell}\in \mathcal{S}(H)$ be such that $S_{\ell}\subseteq S(h(V^{\ell}_{\tau}))$, $\ell \in L_{\tau}$.	Observe that for every $\ell \in L_{\tau^{(1)}}\cap L_{\tau^{(2)}}$, $V^{\ell}_{\tau} \supseteq V^{\ell}_{\tau^{(i)}}$, so $S_{\ell}\subseteq S(h(V^{\ell}_{\tau^{(i)}}))$, $i=1,2$. By definition, there exists \(p_i \in P_{\tau^{(i)}}\) such that $p_i(\ell) = S_{\ell}$ for $\ell \in L_{\tau^{(i)}}$, $i=1,2$.
	Then for \(p =p_1\cup p_2\) we have \(p(\ell) = S_{\ell}\), $\ell \in L_{\tau}$.
	
	For the converse, fix  some \(p \in P_\tau\).
	Let \(p_1 \in P_{\tau^{(1)}}\), \(p_2 \in P_{\tau^{(2)}}\) be functions such that \(p =p_1\cup p_2\). Let
	\(h_i\) be a witness of $p_i$ in $\tau^{(i)}$, $i=1,2$.
	We define $h=h_1\cup h_2$. Since $G_{\tau}$ doesn't contain edges between $V( G_{\tau^{(1)}})$ and $V( G_{\tau^{(2)}})$, $h$ is a homomorphism from \( G_{\tau}\) to \(H\). For all \(\ell \in L_{\tau^{(1)}} \setminus L_{\tau^{(2)}}\), we have \(p(\ell) = p_1(\ell) \subseteq S(h_1(V^{\ell}_\tau))=S(h(V^{\ell}_\tau))\), similarly for  \(\ell \in L_{\tau^{(2)}} \setminus L_{\tau^{(1)}}\). For \(\ell \in L_{\tau^{(1)}} \cap L_{\tau^{(2)}}\), we have \(p(\ell) = p_i(\ell)\subseteq S(h_i(V^{\ell}_{\tau^{(i)}})) \), $i=1,2$, so $p(\ell)\subseteq S(h_1(V^{\ell}_{\tau^{(1)}}))\cap S(h_2(V^{\ell}_{\tau^{(2)}})) =S(h(V^{\ell}_\tau))\). Hence $h$ is a witness of $p$ in $\tau$.  
	
	\paragraph*{\(\tau = \eta_{i, j}(\tau')\) and \(P_{\tau'}\) has already been computed.}
	In this case \(L_\tau = L_{\tau'} \setminus I\) where \(I \subseteq \{i,j\}\) is the set of live labels in \(\tau'\) that are no longer live labels in \(\tau\).
	We set $P_{\tau}$ equal to 
$$\{p:L_{\tau}\to \mathcal{S}(H)~|~ \exists p' \in P_{\tau'}\colon p'(i) \supseteq S(p'(j)) ~\wedge   p|_{{L_\tau}\setminus \{i,j\}}=p'|_{{L_\tau}\setminus \{i,j\}}~\wedge~ p(i)\subseteq p'(i) ~\wedge~ p(j)\subseteq p'(j)\}.$$
	
Intuitively, we can add the edges between two live labels if and only if there are edges between their images in $H$. Our restriction on $p'$ is an expression of this condition in terms of images of neighbors and their signatures.
	Indeed, for correctness, let \(h\) be a homomorphism from \( G_{\tau}\) to \(H\) and $S_{\ell}\in \mathcal{S}(H)$ be such that $S_{\ell}\subseteq S(h(V^{\ell}_{\tau}))$, $\ell \in L_{\tau}$. There exists \(p' \in P_{\tau'}\) such that $p'(\ell)=S_{\ell}$ for all $\ell \in L_{\tau}\setminus \{i,j\}$,  $p'(i)=S(h(V^{i}_{\tau'}))$ and  $p'(j)=S(h(V^{j}_{\tau'}))$. As \(h\) is a homomorphism, we have $h(V^i_{\tau'})\times h(V^j_{\tau'}) \subseteq E(H)$, which means that $S(h(V^j_{\tau'}))\supseteq h(V^i_{\tau'})$, i.e. $p'(j)\supseteq h(V^i_{\tau'})$. 
	Then $S(p'(j)) \subseteq S(h(V^i_{\tau'}))=p'(i)$ and hence \(p'\) gives rise to \(p \in P_\tau\) 
        such that  $p(\ell)=S_{\ell}$, $\ell \in L_{\tau}$.
	
	On the other hand, let \(p \in P_\tau\) arise from \(p'\in P_{\tau'}\).
	Consider a witness \(h \colon  G_{\tau'} \to H\) of $p'$ in $\tau'$. To see that $h$ preserves edges between $V^i_{\tau'}$ and $V^j_{\tau'}$, recall
that $p'(i) \supseteq S(p'(j))$, so $S(h(V^i_{\tau'}))\supseteq p'(i) \supseteq S(p'(j)) \supseteq S(S(h(V^j_{\tau'}))) \supseteq h(V^j_{\tau'})$. Hence $h(V^i_{\tau'})\times h(V^j_{\tau'}) \subseteq E(H)$, so $h$ is a homomorphism. By construction, for every $\ell \in L_{\tau}$ it holds that $p(\ell)\subseteq p'(\ell)\subseteq S(h(V^{\ell}_{\tau}))$. Therefore $h$ witnesses $p$ in $\tau$.
	
It is easy to verify that \(|P_\tau| \le s(H)^{\cw{G}}\) for each subexpression \(\tau\) of \(\sigma\).
	This means that in each step, the computation requires time  \(\bigoh(\cw{G}  s(H)^2 s(H)^{\cw{(G)}})\).
	Overall this yields a complexity of \(\mathcal{O}(|V(G)|\cw{G} s(H)^2 s(H)^{\cw{(G)}}) \subseteq \mathcal{O}^*(s(H)^{\cw{G}})\).
\end{proof}

\section{On Products and Projectivity}

While Theorem~\ref{thm:algorithm-main} will serve as the upper bound that will match our target SETH-based lower bounds for \homo{H} for the ``most difficult'' choices of $H$, in many cases one can in fact supersede the algorithm's runtime by exploiting well-known properties of target graphs. 

As a simple example showcasing this, consider the \emph{wheel} graph $W_6$ (see Figure~\ref{fig:homo-signature}).
Since $W_6$ is $3$-colorable, it holds that $W_6 \to K_3$, and since $K_3$ is a core and an induced subgraph of $W_6$, it is the core of $W_6$.
We recall that if $H$ is a core of $H'$, then for every graph $G$ it holds that $G \to H$ if and only if $G \to H'$.
Hence, having an instance $G$ of \homo{W_6}, we can compute a core of $W_6$ (since we assume that the target graph is 
fixed, this can be done in constant time), and use \Cref{thm:algorithm-main} for $H=K_3$ to decide whether $G \to W_6$ in total running time $\Ohs(s(K_3)^{\cw{G}})$.
As $s(K_3)<s(W_6)$ (as showcased in Figure~\ref{fig:homo-signature}), this yields a better running time bound than the direct use of \Cref{thm:algorithm-main}.
While this example shows that the signature number can decrease by taking an induced subgraph, we remark that it can never increase.

\begin{figure}
  \begin{minipage}[c]{0.4\textwidth}
\begin{tabular}{ccc}
\begin{tikzpicture}[every node/.style={draw,circle,fill=white,inner sep=0pt,minimum size=5pt}]
\def\n{6}
\foreach \i in {1,...,\n}
{
	\draw[line width=1] (360/\n*\i-360/\n:1) -- (360/\n*\i:1);
	\draw[line width=1] (0:0) -- (360/\n*\i:1);
}

\foreach \i in {1,...,\n}
{
	\node[label={360/\n*\i:\i}] (a\i) at (360/\n*\i:1) {};	
}
\node[label=0, fill=yellow] (a) at (0:0) {};	
\node[fill=blue] at (0:1) {};	
\node[fill=red] at (60:1) {};	
\node[fill=blue] at (120:1) {};	
\node[fill=red] at (180:1) {};	
\node[fill=blue] at (240:1) {};	
\node[fill=red] at (300:1) {};	
\end{tikzpicture}
&
\begin{tikzpicture}
\node at (0,0) {$\to$};
\node at (0,-1) {};
\end{tikzpicture}
&
\begin{tikzpicture}[scale=0.6,every node/.style={draw,circle,fill=white,inner sep=0pt,minimum size=5pt}]
\def\n{3}
\foreach \i in {1,...,\n}
{
	\draw[line width=1] (360/\n*\i-360/\n:1) -- (360/\n*\i:1);
}
\foreach \i in {1,...,\n}
{
	\node (a\i) at (360/\n*\i:1) {};
	}	
\node[fill=yellow]at (0:1) {};		
\node[fill=blue] at (120:1) {};	
\node[fill=red] at (240:1) {};
\node[draw=none,fill=none] at (0,-1.8) {};
\end{tikzpicture}
\end{tabular}
  \end{minipage}\hfill
  \begin{minipage}[c]{0.58\textwidth}
    \caption{The graphs $W_6$ (left) and $K_3$ (right).
Colors on the vertices of $W_6$ indicate the homomorphism $h: W_6 \to K_3$. 
We note that $\{\{0\}\} \cup \{\{0,i\}~|~i \in [6]\} \subseteq \cS(W_6)$.
Since $K_3$ is a clique, we have $s(K_3)=6 < 7 \leq s(W_6)$.}
\label{fig:homo-signature}
  \end{minipage}
\end{figure}
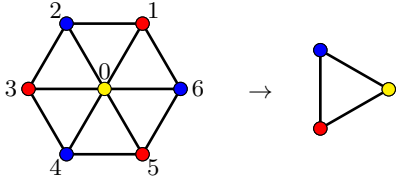

\begin{observation}
\label{obs:subgraph-s-bound}
Let $H$ and $H'$ be graphs such that $H$ is an induced subgraph of $H'$. 
Then $s(H) \leq s(H')$.
\end{observation}
\begin{proof}
Given a connected graph $Q$ without loops, one may consider an equivalence relation $\sim_Q$ on the set of nonempty subsets of $V(Q)$ defined as follows: $V_1\sim_Q V_2$ if and only if $V_1$ and $V_2$ have the same signature sets in $Q$. Observe that $s(Q)$ is equal to the number of equivalence classes of 
$\sim_Q$ minus one (as there are subsets $V$ such that $S(V)=\emptyset$). Hence to prove the claim, it suffices to show that whenever two subsets of $V(H)$ belong to different equivalence classes of $\sim_{H}$, they also belong to different equivalence classes of $\sim_{H'}$. For this, consider any two non-empty subsets $V_1$ and $V_2$ of $V(H)$ such that $V_1 \nsim_H V_2 $. Without loss of generality, we assume that there exists $v\in (\bigcap_{t\in V_1}N_H(t)) \setminus (\bigcap_{t\in V_2}N_H(t))$. Then $vt\in E(H) \subseteq E(H')$ for every $t\in V_1$, i.e., $v$ belongs to the signature set of $V_1$ in $H'$. On the other hand, $vt_0\notin E(H)$ for some $t_0\in V_2$. As $H$ is induced subgraph of $H'$, it means that $vt_0\notin E(H')$, so $v$ doesn't belong to the signature set of $V_2$ in $H'$ and hence $V_1 \nsim_{H'} V_2 $.
\end{proof}

At this point, we may ask whether the procedure of simply computing the unique core $H$ of the fixed target $H'$ and then applying~\Cref{thm:algorithm-main} for $H$ could yield a tight upper bound for \homo{H'}. Unfortunately, the situation is more complicated than that, and we need to introduce a few important notions in order to capture the problem's fine-grained complexity.

Let the \emph{direct product} $H_1 \times H_2$ of graphs $H_1, H_2$ be the graph defined as follows:
\begin{align*}
V(H_1 \times H_2) &= V(H_1) \times V(H_2), \\
E(H_1 \times H_2) &= \{(x_1,x_2)(y_1,y_2): x_iy_i \in E(H_i) \textrm{ for every } i \in \{1,2\}\}.
\end{align*}
We call $H_1$ and $H_2$ the \emph{factors} of $H_1 \times H_2$.
Clearly, the operation $\times$ is commutative, and since it is also associative, we can naturally extend the definition of direct product to more than two factors, i.e., $H_1 \times H_2 \times\ldots \times H_m = H_1 \times (H_2 \times\ldots \times H_m)$.
Note that for every graph $H$ it holds that $H \times K^*_1=H$..

In the remaining part of the paper we will often consider vertices that are tuples.
If such a vertex is an argument of some function and in cases where this does not lead to confusion, we omit one pair of brackets; similarly, we omit internal brackets in nested tuples where this does not lead to confusion. 
Moreover, for any graph $H$ and for an integer $\ell$, we denote by $H^\ell$ the graph $\overbrace{H \times \ldots \times H}^\ell$.
As an example, instead of writing $((x_1,x_2),y_1) \in ((H_1\times H_1) \times H_2)$, we write $(x_1,x_2,y_1) \in (H_1^2 \times H_2)$.

If $H=H_1\times\ldots\times H_m$ for some graphs $H_1,\ldots,H_m$, we say that $H_1\times\ldots\times H_m$ is a \emph{factorization} of $H$.
A graph $H$ on at least two vertices is \emph{prime} if the fact that $H=H_1 \times H_2$ for some graphs $H_1, H_2$ implies that $H_1=K^*_1$ or $H_2=K^*_1$.
If $H$ has a factorization $H_1 \times \ldots \times H_m$ such that for every $i \in [m]$ the graph $H_i$ is prime, we call $H_1 \times \ldots \times H_m$ a \emph{prime factorization} of $H$.

\begin{theorem}[\hspace{-0.001cm}\cite{dorfler1974primfaktorzerlegung,hammack2011handbook}] \label{thm:prime-factorization}
Any connected non-bipartite graph with more than one vertex has a unique
prime factorization (into factors with possible loops). 
\end{theorem}

Consider a graph $H_1 \times \ldots \times H_m$, and let $i \in [m]$. 
A mapping $\pi_i: V(H_1 \times \ldots \times H_m) \to V(H_i)$ such that $\pi_i(x_1, \ldots, x_\ell)=x_i$ is called the ($i$-th) \emph{projection} of $H_1 \times \ldots \times H_m$.
Clearly, such a mapping is always a homomorphism.

\begin{observation}
\label{obs:homo-products} 
Let $G,H_1, \ldots, H_m$ be graphs.
Then $G \to H_1 \times \ldots \times H_m$ if and only if for every $i \in [m]$ we have $G \to H_i$. 
\end{observation}
\begin{proof}
Let $f: G \to H_1 \times \ldots \times H_m$.
Then for every $i \in [m]$ we have a homomorphism $\pi_i: G \to H_i$. 
Conversely, if for every $i \in [m]$ we have $g_i: G \to H_i$, then we can define $g:G \to H_1 \times \ldots \times H_m$ as $g(x)=(g_1(x),\ldots,g_m(x))$.
\end{proof}
Crucially, since there exist cores that are not prime~\cite{OkrasaSODAJ}, in some cases Observation~\ref{obs:homo-products} allows us to improve the bounds given by Theorem~\ref{thm:algorithm-main} even if $H$ is a core, simply by considering all possible factorizations of $H$.

\begin{corollary}
\label{cor:algorithm-factors}
Let $H$ be a graph with factorization $H_1 \times \ldots \times H_m$, and let $G$ be an instance graph of \homo{H}.
Assuming that the clique-width expression \(\sigma\) of $G$ of width $\cw{G}$ is given, the \homo{H} problem can be solved in time $\max_{i \in [m]}\Ohs\big(s(H_i)^\cw{G}\big)$.
\end{corollary}
\begin{proof}
Observe that if $G$ is an instance of \homo{H}, by \Cref{thm:algorithm-main} for every $i \in [m]$ we can decide whether $G \to H_i$ in time $\Ohs\big(s(H_i)^\cw{G}\big)$.
Then, if $G$ is a yes-instance of \homo{H_i} for every $i \in [m]$, we return that $G$ is a yes-instance of \homo{H}. 
Otherwise, we return that $G$ is a no-instance of \homo{H}.
The correctness of this procedure follows from \Cref{obs:homo-products}.
\end{proof}

On the other hand, the notion of signature sets we introduced in the previous section behaves multiplicatively with respect to taking direct product of graphs. 

\begin{observation}
\label{obs:s1-times-s2}
Let $H=H_1 \times H_2$.
Then $\mathcal{S}(H)=\mathcal{S}(H_1) \times \mathcal{S}(H_2)$.
\end{observation}
\begin{proof}
We prove that $\mathcal{S}(H)$ is of form $\{S(T_1) \times S(T_2): T_i \subseteq V(H_i), S(T_i) \neq \emptyset \textrm{ for } i=1,2\}$.

Let $T_1$ and $T_2$ be some subsets of, respectively, $V(H_1)$ and $V(H_2)$.
Clearly,  
\begin{align}\label{eq:st1st2}
S(T_1) \times S(T_2) = [\bigcap_{t \in T_1} N(t)] \times [\bigcap_{t' \in T_2} N(t')] = \bigcap_{(t,t') \in T_1 \times T_2} N(t,t') = S(T_1 \times T_2).
\end{align}
Therefore, if $S(T_1)$ and $S(T_2)$ are non-empty, we get that $S(T_1) \times S(T_2) \in \mathcal{S}(H).$

To see that $\mathcal{S}(H) \subseteq \mathcal{S}(H_1) \times \mathcal{S}(H_2)$, we show that for every $T \subseteq V(H)$ set $S(T)$ is of the form $S(T_1) \times S(T_2)$ for some $T_1,T_2$.
Define $T_1$ and $T_2$ to be minimal sets such that $T \subseteq T_1 \times T_2$.
Hence, by \eqref{eq:st1st2}, $S(T_1) \times S(T_2) = S(T_1 \times T_2) \subseteq S(T)$. 
On the other hand, for every $(s,s') \in S(T)$ we have $s \in \bigcap_{t \in T_1} N(t)$ and $s' \in \bigcap_{t' \in T_2} N(t')$, so the equality follows.
\end{proof}

In particular, it follows from \Cref{obs:s1-times-s2} that if $H$ is a graph with factorization $H_1 \times \ldots \times H_m$, then $s(H)=s(H_1) \cdot \ldots \cdot s(H_m)$.
Therefore if there exist at least two factors $H_i, H_j$ such that $s(H_i)>1$, $s(H_j)>1$, \Cref{cor:algorithm-factors} yields a better running time than \Cref{thm:algorithm-main}.

\medskip

In order to analyze the possible matching lower bounds for our algorithms, 
in the remaining part of the section, we focus only on connected non-trivial cores $H$ that are provided with their unique prime factorization $H_1 \times \ldots \times H_m$; if $H$ is prime, we technically consider this factorization to be $H\times K_1^*$ (noting that this is not a prime factorization, and that $K_1^*$ is the only non-simple graph in this article).
We note that the factors of a core must satisfy some necessary conditions.

\begin{observation}[\hspace{-0.001cm}\cite{OkrasaSODAJ}]\label{obs:core-factors}
Let $H$ be a connected, non-trivial core with factorization $H = H_1 \times \ldots \times H_m$ such that $H_i \neq K^*_1$ for all $i \in [m]$.
Then for every $i \in [m]$ the graph $H_i$ is a connected non-trivial core, incomparable with $H_j$ for $j \in [m] - \{i\}$.
\end{observation}

\Cref{obs:core-factors} in particular implies that if $H$ is a connected non-trivial graph with factorization $H_1 \times \ldots \times H_m$, then at least one of the factors $H_i$ must be non-trivial, and that $K_1$ and $K_2$ never appear as factors of a connected non-trivial graph.

In the remaining part of this section we introduce a few more important definitions, in particular, the well-established notion of \emph{projectivity} for non-trivial graphs $H$.

We say that a homomorphism $f: H^\ell \to H$, for some $\ell \geq 2$, is \emph{idempotent} if for each $x \in V(H)$ it holds that $f(x,\ldots,x)=x$.
Graph $H$ is \emph{projective} if for every $\ell \geq 2$, every idempotent homomorphism $f: H^\ell \to H$ is a projection.
We note that every projective graph on at least three vertices must be connected, ramified, non-bipartite and prime~\cite{LaroseT01}.

Here, we introduce a generalization of the projectivity property for non-trivial cores, which turns out to be the central component required to establish the lower bound for our problem. As a first step towards this, we lift the notion of idempotency as follows. Let $H$ be a connected, non-trivial prime core, and let $W$ be either a connected core on at least three vertices incomparable with $H$, or the graph $K^*_1$.
Let $f: A\rightarrow H$ be a homomorphism where $A=H^\ell\times W$; observe that $\ell$ is uniquely determined by either $W$ being incomparable with the prime core $H$, or $W$ being $K^*_1$.
We say that $f$ is \emph{$H$-idempotent} if for each $x \in V(H), y \in V(W)$ it holds that $f(x,\ldots,x,y)=x$.

Now, let us consider a non-trivial core $H$ which admits a prime factorization $H_1\times \ldots \times H_m$ and let $i\in [m]$. We say that $H$ is \emph{$H_i$-projective} if $H_i$ is non-trivial and every $H_i$-idempotent homomorphism $f:H_1 \times \ldots \times H_{i-1} \times H_i^\ell \times H_{i+1} \times \ldots \times H_m \to H_i$ is a projection.
In other words, for every homomorphism $f:H_1 \times \ldots H_{i-1} \times H_i^\ell \times H_{i+1} \times \ldots \times H_m \to H_i$ such that for every $x \in V(H_i), y_j \in V(H_j)$ for $j \in [m] - \{i\}$ it holds that $f(y_1,\ldots,y_{i-1},x,\ldots,x,y_{i+1},\ldots,y_m)=x$, we must have that there exists $q \in \{i,\ldots,i+\ell-1\}$ such that $f \equiv \pi_q$. 
Recall that if $H$ is a non-trivial projective core, then it must be prime, so $H \times K_1^*$ is its only possible factorization.
It is straightforward to verify that in a such case $H$ is $H$-projective.

Since the direct product of graphs is commutative, if $H=H_1 \times \ldots \times H_m$ is $H_i$-projective for some $i \in [m]$, to simplify the notation we will often assume w.l.o.g.\ that $i=1$.

\section{Hardness}
In this section, we focus on establishing the desired lower bounds, stated below.

\thmlowerbound*

We divide our proof into two main steps.
First, we show that in our setting, instead of considering the \homo{H} problem, we may focus on the \textsc{Homomorphism Extension} problem, denoted \hcoloringext{H}.
For a fixed $H$, \hcoloringext{H} takes as an instance a pair $(G',h')$, where $G'$ is a graph and $h':V' \to V(H)$ is a mapping from some $V' \subseteq V(G')$.
We ask whether there exists \emph{an extension} of $h'$ to $G'$, i.e., a homomorphism $h: G' \to H$ such that $h|_{V'} \equiv h'$.

The \hcoloringext{H} is clearly a generalization of the \homo{H} problem.
However, as the first step of our proof, we show that if $H$ is a fixed non-trivial core, each instance $(G',h')$ of \hcoloringext{H} can be transformed in polynomial time into an instance $G$ of \homo{H}, such that $\cw{G'}$ and $\cw{G}$ differ only  by a constant.

\begin{theorem}\label{thm:homoextiffhomo}
Let $H$ be a fixed non-trivial core.
Given an instance $(G',h')$ of \hcoloringext{H}, we can construct an equivalent instance $G$ of \homo{H} such that $\cw{G}\leq\cw{G'}+|V(H)|$.
\end{theorem} 
\begin{proof}
Let $V' \subseteq V(G')$ be the domain of $h'$.
We construct $G$ by taking a copy $\hat{G'}$ of $G'$ and a copy $\hat{H}$ of $H$.
Then, for every $v \in V'$ we add all the edges with one endpoint in $v$ and another one in $N_{\hat{H}}(h'(v))$.

Observe that if there exists an extension $h:G' \to H$ of $h'$, then $h$ can be also extended to $G$, by setting $h(v)=v$ for every $v \in V(\hat{H})$.
Indeed, let $uv \in E(G)$. 
If $u,v \in V(\hat{G'})$, then $h(u)h(v) \in E(H)$, by definition of the extension of $h'$ to $G'$. 
If $u,v \in V(\hat{H})$, then $h(u)h(v)=uv \in E(H)$.
Finally, assume that $u \in V(\hat{G'}), v \in V(\hat{H})$.
Note that, by definition of $G$, this can happen only if $u \in V'$ and $v$ is adjacent to $h'(u)$.
Hence, $h(u)h(v) = h'(u)v \in E(H)$.

For the reverse direction, assume that there exists a homomorphism $f: G \to H$. We show that there exists an extension $h: G' \to H$ of $h'$. 
Let $\sigma: \hat{H} \to H$ be a restriction of $f$ to $\hat{H}$. 
Since $H$ is a core, $\sigma$ is an automorphism of $H$.
We claim that $g = \sigma^{-1} \circ f: G \to H$, restricted to $G'$, is an extension of $h'$.
Clearly, $g$ is a composition of homomorphisms, so also a homomorphism.
Therefore, it remains to show that for every $v \in V'$ we have $h'(v)=g(v)$. 
Since $H$ is a core, and $N_{\hat{H}}(h'(v)) \subseteq N_G(v)$, we have that $f(v) = f(h'(v))$.
It follows that $g(v) = \sigma^{-1} \circ f(v) = \sigma^{-1} \circ f(h'(v)) = g(h'(v))$.
However, recall that for every $u \in V(H)$ we have that $g(u)=u$, so in particular, $g(h'(v))=h'(v)$.

To see that $\cw{G}\leq\cw{G'}+|V(H)|$, observe that we added exactly $|V(H)|$  vertices to \(G'\).
This means we can modify a clique-width expression \(\sigma\) for \(G'\) to obtain a clique-width expression of \(G\) as follows.
Each added vertex is introduced with a designated label that is distinct from all labels used in \(\sigma\).
Then each subexpression of \(\sigma\) that introduces a vertex of \(G'\) can be replaced by an expression that introduces  the vertex and inserts all required edges to the added vertices.
Finally, one can insert the missing edges between added vertices.
\end{proof}

As the second step, we prove the following theorem.

\begin{theorem}\label{thm:lower-bound-homoext}
Let $H$ be a fixed non-trivial core with prime factorization $H_1 \times \ldots \times H_m$.
Assume that $H$ is $H_i$-projective for some $i \in [m]$.
Then there is no algorithm solving \hcoloringext{H} in time $\Ohs((s(H_i)-\epsilon)^{\cw{G'}})$ for any $\epsilon >0$, unless the SETH fails.
\end{theorem}

Before we proceed to the proof of \Cref{thm:lower-bound-homoext}, we show that it implies \Cref{thm:lower-bound-homo}.
\\
\textbf{\Cref{thm:lower-bound-homoext} $\to$ \Cref{thm:lower-bound-homo}:} 
Let $H$ be a non-trivial core with a prime factorization $H_1 \times \ldots \times H_m$.
W.l.o.g. assume that $H$ is $H_1$-projective.
Suppose that \Cref{thm:lower-bound-homo} does not hold, i.e., there exists an algorithm $A$ that solves every instance $G$ of \homo{H} in time $\Ohs((s(H_1)-\epsilon)^{\cw{G}})$.

Let $(G',h')$ be an instance of \hcoloringext{H}.
We use \Cref{thm:homoextiffhomo} to transform $(G',h')$ into an equivalent instance $G$ of \homo{H}, such that $\cw{G}\le \cw{G'}+|V(H)|$.
Then, we use $A$ to decide whether $G \to H$ in time 
\[\Ohs\left((s(H_1)-\epsilon)^{\cw{G}}\right)=\Ohs\left((s(H_1)-\epsilon)^{\cw{G'}}\cdot (s(H_1)-\epsilon)^{|V(H)|}\right).\] 
Since $H$ is a fixed graph, $(s(H_1)-\epsilon)^{|V(H)|}$ is a constant, and therefore $\Ohs\left((s(H_1)-\epsilon)^{\cw{G}}\right) = \Ohs\left((s(H_1)-\epsilon)^{\cw{G'}}\right)$.
Since $G \to H$ if and only if $(G',h')$ is a yes-instance of \hcoloringext{H}, we get a contradiction with \Cref{thm:lower-bound-homoext}. \qedhere
\medskip
We will prove \Cref{thm:lower-bound-homoext} for $i=1$, which covers other cases by commutativity of direct products. We begin by constructing certain gadgets that will be used later.
Let $H$ be a fixed core with factorization $H_1 \times \ldots \times H_m$.
We define $W=H_2 \times \ldots \times H_m$ if $m \geq 2$, and $W=K^*_1$ otherwise.
Clearly, $H_1 \times W$ is a (not necessarily prime) factorization of $H$.
Moreover, if for some graph $G$ we have a homomorphism $f: G \to H_1 \times \ldots \times H_m$, for $i \in [m]$ we denote by $f_i$ the homomorphism $\pi_i \circ f : G \to H_i$.
Let $S$ be a set of pairs of vertices of $H_1$, and let $w,w' \in V(W)$.
We say that a tuple $(F,h',p,q)$, such that $F$ is a graph, $h': V' \to H$ is a mapping with domain $V' \subseteq V(F)$, and $p,q \in V(F)$, is an \emph{$(S,w,w')$-gadget} if 
\begin{enumerate}[(S1)]
    \item \label{lab:notins} for every extension $h: F \to H$ of $h'$, it holds that $(h_1(p),h_1(q))\in S$,

    \item \label{lab:ins} for every pair $(s_1,s_2) \in S$ there exists an extension $h:F \to H$ of $h'$ such that $h(p)=(s_1,w)$ and $h(q)=(s_2,w')$.
\end{enumerate}

\begin{lemma}\label{lem:non-projective-s-gadget}
Let $H$ be a non-trivial connected core with factorization $H_1 \times W$, let $S \subseteq V(H_1)^2$, and let $w, w' \in V(W)$.
Assume that $H$ is $H_1$-projective.
Then there exists an $(S,w,w')$-gadget.
\end{lemma}
\begin{proof}
Let $S=\{(s^1_1,s^1_2), \ldots, (s^\ell_1,s^\ell_2)\}$.
Define
\begin{align*}
F=H_1^\ell \times W, \ \textrm{ and } \ p=(s^1_1,\ldots,s^\ell_1,w), \ \textrm{ and } q=(s^1_2,\ldots,s^\ell_2,w').
\end{align*}
Let $V'=\{(x,x,\ldots,x,y)~|~ x \in V(H_1), y \in V(W)\}$, and let $h'(x,\ldots,x,y)=(x,y)$.  
We claim that $(F,h',p,q)$ is an $(S,w,w')$-gadget.

The condition \ref{lab:notins} follows from the fact that $H$ is $H_1$-projective.
Indeed, if $h: F \to H_1 \times W$ is an extension of $h'$, observe that $h_1$ must be $H_1$-idempotent, and hence a projection on one of the $\ell$ first coordinates.
Therefore, we must have $(h_1(p),h_1(q)) \in S$.

For (S2), take any $(s^i_1, s^i_2) \in S$ and let $h: F \to H_1 \times W$, $h(x)= (\pi_i(x), \pi_{\ell+1}(x)).$
Clearly, $h$ is an extension of $h'$, and it is easy to verify that $h(p)=(s^i_1,w)$ and $h(q)=(s^i_2,w')$.
\end{proof}
We say that $S \subseteq V(H_1)^2$ is \emph{proper}, if for every coordinate there exist two elements in $S$ that differ on that coordinate, i.e., $S$ is not of the form $\{s\} \times U$ nor $U \times \{s\}$ for some $s \in V(H_1)$ and $U \subseteq V(H_1)$.
Note that if $S$ is proper and $(F,h',p,q)$ is an $(S,w,w')$-gadget constructed as in \Cref{lem:non-projective-s-gadget}, then neither $p$ nor $q$ belong to the domain of $h'$.
For fixed vertices $a,b \in V(H_1)$, let $S_{a,b}=\{(a',b'): a' \neq a, b' \in V(H_1)\} \cup \{(a,b)\}$.
We call the $(S_{a,b},w,w')$-gadget $(F,h',p,q)$ an \emph{$((a,b),w,w')$-implication-gadget}.
Intuitively, an $((a,b),w,w')$-implication-gadget works as the implication $a \Rightarrow b$, since in every homomorphism $h: F \to H$ that extends $h'$, if $h_1(p)=a$, then $h_1(q)=b$. 
Let $a,b,c \in V(H_1)$, $w \in V(W)$, and let $t$ be an integer. 
A triple $(F,h',R)$ such that $F$ is a graph, $h': V' \to H_1 \times W$ is a partial mapping from some $V' \subseteq V(F)$, and $R$ is a subset of $V(F)$ of cardinality $t$ is an \emph{$t$-or-gadget with domain $((a,b,c),w)$} if 
\begin{enumerate}[(O1)]
    \item \label{lab:or-gadget-a} for every homomorphism $h:F \to H$ that is an extension of $h'$, and for every $u \in R$ we have that $h_1(u) \in \{a,b,c\}$ and there exists $v \in R$ such that $h_1(v)=a$,
    \item \label{lab:or-gadget-nota} for every $v \in R$ there exists a homomorphism $h: F \to H$ that is an extension of $h'$, such that $h(v)=(a,w)$ and for every $u \in R - \{v\}$ it holds that $h(u) \in \{(b,w),(c,w)\}$.
\end{enumerate}

\begin{lemma}\label{lem:non-projective_or_gadget} 
Let $H$ be a non-trivial core with factorization $H_1 \times W$.
Assume that $H$ is $H_1$-projective.
Then for every distinct $a,b,c \in V(H_1)$, every $w \in V(W)$ and every $t$, there exists a $t$-or-gadget $(F,h',R)$ with domain $((a,b,c),w)$.
\end{lemma}
\begin{proof} 
We consider separately the cases $t=1$ and $t=2$.
Observe that in case $t=1$ our gadget needs to be a graph that has a vertex $r \in R$ that is always mapped to $(a,w)$. 
Hence, we set $F=K_1$, $R=V(F)$, and $h'(v)=(a,w)$ for $v \in V(F)$.

If $t=2$, let $S=\{(a,b),(b,a),(a,a)\}$, we introduce an independent set $R=\{r_1,r_2\}$ and $(S,w,w)$-gadget $(F,h',r_1,r_2)$. To see that $(F,h',R)$ satisfies \ref{lab:or-gadget-a}, consider any extension $f: F \to H$ of $h'$. As $(F,h',r_1,r_2)$ is $(S,w,w)$-gadget, we have $(f_1(r_1), f_1(r_2))\in \{(a,b),(a,a),(b,a)\}$. For \ref{lab:or-gadget-nota}, recall that by the property \ref{lab:ins} of $S$-gadget there exist extensions $f^{(1)}$ and $f^{(2)}$ of $h'$ such that $(f^{(1)}_1(r_1), f^{(1)}_1(r_2))=(a,b)$ and $(f^{(2)}_1(r_1), f^{(2)}_1(r_2))=(b,a)$.

Assume then that $t > 2$, and let 
\begin{align*}
S&={\{a,b,c\}}^2 - \{(b,c),(c,b)\}, \\
S_{\text{left}}&=\{(a,a), (a,b), (a,c), (c,a), (c,c)\}, \\ 
S_{\text{right}}&=\{(a,a), (a,b), (b,a), (b,b), (c,a)\}
\end{align*}
be subsets of $V(H_1)^2$.
We introduce an independent set $R=\{r_1,\ldots,r_t\}$ of $t$ vertices and one copy of $(S_{\text{left}},w,w)$-gadget $(F_1,h'_1,r_1,r_2)$.
Then, for $j \in \{2,\ldots,t-2\}$, we introduce an $(S,w,w)$-gadget $(F_j,h'_j,r_j,r_{j+1})$ (we note that if $t= 3$, we do not introduce these).
Last, we introduce one copy of $(S_{\text{right}},w,w)$-gadget $(F_{t-1},h'_{t-1},r_{t-1},r_t)$.
We note that sets $S, S_{\text{left}}$, and $S_{\text{right}}$ are proper, so the domains of the partial mappings $h'_j$, $j \in \{1,\ldots,t-1\}$, are pairwise disjoint. In particular, the union $h' = \bigcup_{j=1}^t h'_j$ is a well-defined mapping.
We define $F$ to be the union of all the graphs from the introduced gadgets and claim that $(F,h',R)$ is a $t$-or-gadget.

We first show that \ref{lab:or-gadget-a} holds.
Assume that there exists an extension $f: F \to H$ of $h'$, and $j' \in [t]$ such that $f_1(r_{j'}) \notin \{a,b,c\}$. 
This implies that there exists $j \in \{j'-1, j'\}$ such that $(f_1(r_{j}),f_1(r_{j+1})) \notin S'$ for any $S' \in \{S_{\text{left}}, S, S_{\text{right}}\}$. This is a contradiction with $(F_,h'_j,p_j,q_j)$ being an $(S',w,w)$-gadget, as it violates \ref{lab:notins}. 

Now assume that there exists an extension $f: F \to H$ of $h'$ such that for every $j \in [t]$ we have that $f_1(r_j) \in \{b,c\}$.
The definition of $S_{\text{left}}$ and $S_{\text{right}}$, respectively, implies that $f_1(r_1)=c$ and $f_1(r_t)=b$.
Hence, there exists $j \in [t-1]$ such that $f_1(r_j)=c$ and $f_1(r_{j+1})=b$.
However, observe that the pair $(c,b)$ does not belong to set $S'$, for $S' \in \{S_{\text{left}}, S, S_{\text{right}}\}$, and since we introduced an $S'$-gadget from $r_j$ to $r_{j+1}$, this leads to a contradiction. 

To see that \ref{lab:or-gadget-nota} holds as well, fix some $r_j \in R$ and define
\begin{align*}
f'(r_\ell) =
\begin{cases}
(a,w), & \textnormal{ if } \ell=j, \\
(c,w), & \textnormal{ if } \ell<j, \\
(b,w), & \textnormal{ if } \ell>j,
\end{cases}
\end{align*}
If $j=1$, then since $(a,b) \in S_{\text{left}}$, $(b,b) \in S$ and $(b,b) \in S_{\text{right}}$, the property \ref{lab:ins} asserts that we can construct a homomorphism $f:F \to H$ that extends $h'$ and $f'$.
The same holds also if $j=t$, (since $(c,c) \in S_{\text{left}}$, $(c,c) \in S$ and $(c,a) \in S_{\text{right}}$), and if $1<j<t$ (since $(c,c) \in S_{\text{left}}$, $(c,c),(c,a),(a,b),(b,b) \in S$ and $(b,b) \in S_{\text{right}}$).
\end{proof}

Finally, all that remains is to prove \Cref{thm:lower-bound-homoext}.
Our reduction generalizes the construction used by Lampis~\cite{DBLP:journals/siamdm/Lampis20} to reduce an SETH lower-bounded constraint satisfaction problem to \textsc{\(k\)-Coloring}.
Intuitively speaking, in that construction possible variable assignments are encoded by mapping specified vertices to arbitrary non-trivial subsets of the colors.
The straightforward generalization of this approach to our setting would be to map to non-trivial subsets of \(V(H)\).
However, the structure of \(H\) allows only certain configurations of subsets as images for the specified vertices in a solution for \homo{H}---which is precisely where the signature sets come into play.

Let $q,B \geq 2$ be integers. 
We will reduce from the \qcsp{q}{B} problem that is defined as follows.
An instance of \qcsp{q}{B} consists of a set $X$ of variables and a set $C$ of $q$-constraints.
A $q$-constraint $c \in C$ is a $q$-tuple of elements from $X$ and a set $P(c)$ of $q$-tuples of elements from $[B]$ (i.e., $P(c) \subseteq [B]^q)$.
The \qcsp{q}{B} problem asks whether there exists an assignment $\gamma: X \to [B]$, such that each constraint is satisfied, i.e., if $c=((x_1,\ldots,x_q),P(c)) \in C$, then $(\gamma(x_1),\ldots,\gamma(x_q)) \in P(c)$. 
Note that we can assume that $q$-constraints in our \qcsp{q}{B} instance may have less than $q$ vertices, as it is always possible to add at most $q-1$ dummy variables to $X$ and add them to constraints that are of smaller size.

We will use the following theorem.
\begin{theorem}[\hspace{-0.001cm}\cite{DBLP:journals/siamdm/Lampis20}]\label{thm:q-exists}
For any $B \geq 2, \epsilon >0$ we have the following: assuming the SETH, there exists $q$ such that $n$-variable \qcsp{q}{B} cannot be solved in time $\Ohs((B-\epsilon)^n)$.
\end{theorem} 

We have all the tools to perform the final reduction.

\begin{proof}[Proof of \Cref{thm:lower-bound-homoext}]
Recall that it is sufficient to prove the theorem when $H=H_1 \times W$ is non-trivial $H_1$-projective core ($W=K^*_1$ if $H=H_1$).
Fix $\epsilon >0$ and set $B=s(H_1)$. As $H$ is $H_1$-projective, $H_1$ is non-trivial and hence contains at least three distinct vertices $a$, $b$ and $c$. In particular, $B\geq 3$ by Observation \ref{obs:sh-cores}. Since $H=H_1 \times W$ is a non-trivial core, $W$ must have at least one edge $ww'$ (it may happen that $w=w'$). From now on $a,b,c,w$ and $w'$ are fixed.
Let $q$ be the smallest number such that \qcsp{q}{B} on $n$ variables cannot be solved in time $\Ohs((B-\epsilon)^n)$ assuming the SETH, given by \Cref{thm:q-exists}.

Let $\phi$ be an instance of \qcsp{q}{B}, where $X=\{x_1,\ldots, x_n\}$ is the set of variables and $C=\{c_0,\ldots,c_{m-1}\}$ is the set of constraints.
For every $j \in \{0,\ldots,m-1\}$ denote by $X_j$ the set of variables that appear in the constraint $c_j$.
Let $P(c_{j})=\{f^{j}_1,\ldots,f^{j}_{p_j}\}$ be the set of assignments from $X_{j}$ to $[B]$ that satisfy the constraint $c_{j}$.
Let $L=m(n|H_1|+1)$, and let $\lambda: [B] \to \mathcal{S}(H_1)$ be some fixed bijection. 

We construct the instance $G_\phi$ of \hcoloringext{H}.
For each $j \in \{0,\ldots L-1\}$, let $j'=j \mod m$. 
Let $R_j=\{r^j_1,\ldots,r^j_{p_{j'}}\}$, where each vertex $r^j_k$ corresponds to the assignment $f^{j'}_k$.
We introduce the $p_{j'}$-or-gadget $(F_{j},h'_{j},R_j)$ with domain $((a,b,c),w)$.

For each $x_i \in X_{j'}$, and for each $f^{j'}_k \in P(c_{j'})$ we do the following:
\begin{enumerate}
\item Let $y=f^{j'}_k(x_i) \in [B]$.
Construct an independent set $V^{j,k}_i$ of $|\lambda(y)|$ vertices and an independent set $U^{j,k}_i$ of $|S(\lambda(y))|$ vertices.
\item For each $d \in \lambda(y)$ select a distinct vertex $z \in V^{j,k}_i$ and add an $((a,d),w,w')$-implication-gadget from $r^j_k$ to $z$.
For each $d \in S(\lambda(y))$ select a distinct vertex $z \in U^{j,k}_i$ and add an $((a,d),w,w')$-implication-gadget from $r^j_k$ to $z$.
\item Connect all vertices of $U^{j,k}_i$ with all vertices of previously constructed sets $V^{\ell,k'}_i$ for $\ell <j$ and $k' \in [p_\ell]$  (see Figure~\ref{fig:reduction}).
\end{enumerate}
\begin{figure}[htb]
\centering
\includegraphics[width=\linewidth]{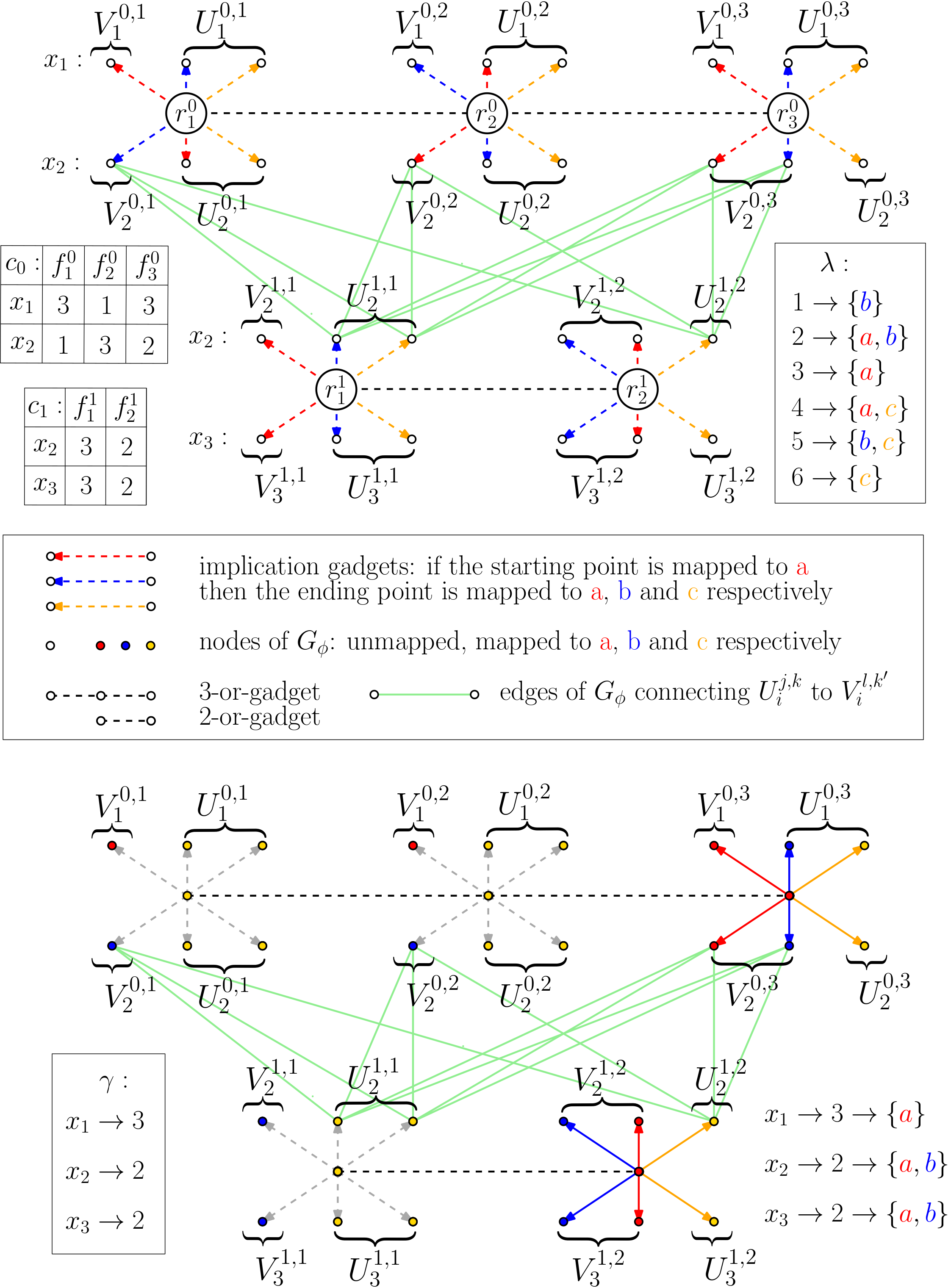}
\caption{Construction of the instance of \hcoloringext{K_3} for $B=6$, variable set $X=\{x_1,x_2,x_3\}$ and constraint set $C=\{c_0,c_1\}$. The sattisfying assignments $f^{j}_k$ and the bijection $\lambda$ are as depicted. The vertices of $K_3$ are $\textcolor{red}{a}, \textcolor{blue}{b}, \textcolor{orange}{c}$. Top: fragment of $G_{\phi}$ restricted to $j=0,1$. Bottom: homomorphism from $G_{\phi}$ to $H$ corresponding to the sattisfying assignment $\gamma$.}
\label{fig:reduction}
\end{figure}

The partial mapping $h'$ is the union of all the partial mappings that are introduced by all the gadgets.
This finishes the construction of the instance $(G_{\phi}, h')$ of \hcoloringext{H}.

\begin{claim}
\label{claim:equiv-forward}
If $\phi$ is a yes-instance of \qcsp{q}{B}, then there exists a homomorphism $h: G_\phi \to H$ that extends $h'$.
\end{claim}
\begin{proof}[Proof of Claim]
If $\varphi$ is a yes-instance of $q$-CSP-$B$, then there exists an assignment $\gamma: X \to [B]$ satisfying each constraint.
We define $h:G_\varphi \to H$ as follows.

Fix $j \in \{0,\ldots, L-1\}$, and consider the or-gadget $(F_j, h'_j,R_j)$.
Recall that the set $P(c_{j'})$ consists of all assignments of variables in $X_{j'}$ that satisfy the constraint $c_{j'}$. 
Therefore, there exists an assignment $f^{j'}_k \in P(c_{j'})$ such that $\gamma|_{X_{j'}} \equiv f^{j'}_k$.
Consider the vertex $r^{j}_k$ that corresponds to that assignment.
By the property \ref{lab:or-gadget-nota} of the or-gadget, we know that there exists a $H$-coloring of $F_j$ that extends $h'$, such that (i) $h_1(r^{j}_k)=a$ and (ii) for every $r^j_{k'} \in R_j, k' \neq k$ we have that $h_1(r^j_{k'}) \in \{b,c\}$.

Let $x_i \in X_{j'}$ and let $y=f^{j'}_k(x_i) \in [B]$.
Since for each $d \in \lambda(y)$ there exists a vertex $z \in V^{j,k}_i$ such that there is an $((a,d),w,w')$-implication-gadget from $r^{j}_k$ to $z$, the condition (i) implies that $h_1(V^{j,k}_i)=\lambda(y)$. 
We color the vertices of $V^{j,k}_i$ in a way that $h(V^{j,k}_i) = \lambda(y) \times \{w\}$.

Also, since for each $d \in S(\lambda(y))$ there exists a vertex $z \in U^{j,k}_i$ such that there is an $(a,d)$-implication-gadget from $r^j_k$ to $z$, the condition (i) implies that $h_1(U^{j,k}_i)=S(\lambda(y))$.
We color the vertices of $U^{j,k}_i$ in a way that $h(U^{j,k}_i) = S(\lambda(y)) \times \{w'\}$.

Because of (ii), the implication gadgets from $r_{k'}$ to the vertices of $V^{j,k'}_i \cup U^{j,k'}_i$ do not put any constraints on the coloring of the sets $V^{j,k'}_i $ and $U^{j,k'}_i$. 
Therefore, for each $v \in V^{j,k'}_i$ we set $h(v)$ to be any vertex from $\lambda(y)$.
Similarly, for each $u \in U^{j,k'}_i$ we set $h(u)$ to be any vertex from $S(\lambda(y))$.
Since $(b,z), (c,z) \in S_{a,d}$ for any $z\in V(H_1)$ and $d \in \lambda(y) \cup S(\lambda(y))$, property~\ref{lab:ins} applied to the implication gadgets asserts 
that since $h_1(r^j_{k'}) \in \{b,c\}$, we can always extend this mapping to a homomorphism of the whole gadget to $H$.

It remains to argue that the edges between the sets $V^{j_1,k_1}_i$ and $U^{j_2,k_2}_i$ are mapped to edges of $H$, for any $j_1<j_2$ and $k_1,k_2$. 
However, observe that since $\gamma$ is an extension of some $f^{j'_1}_{k_1} \in S_{j'_1}$ and $f^{j'_2}_{k_2} \in S_{j'_2}$, we must have $f^{j'_1}_{k_1}(x_i)=f^{j'_2}_{k_2}(x_i)=y$.
Hence, $h$ maps every $v \in V^{j_1,k_1}_i$ to some element of $\lambda(y) \times \{w\}$, and every $u\in U^{j_2,k_2}_i$ to some element of $S(\lambda(y)) \times \{w'\}$.
By Observation \ref{obs:st-mst}, and since $ww' \in E(W)$, we get that $h(v)h(u) \in E(H)$.
That concludes the proof of the claim. 
\end{proof}

\begin{claim}
\label{claim:equiv-backward}
If there exists a homomorphism $h: G_\varphi \to H$ that extends $h'$, then $\varphi$ is a yes-instance of $q$-CSP-$B$.
\end{claim}
\begin{proof}[Proof of Claim]
We will define the assignment $\gamma: X \to [B]$ that makes every constraint from $C$ satisfied.

Fix $j \in \{0,\ldots, L-1\}$, and consider the $p_{j'}$-or-gadget $(F_j,h'_j,R_j)$.
By the property \ref{lab:or-gadget-a} of the or-gadget, there exists $k_j \in [p_{j'}]$ such that $h_1(r^j_{k_j})=a$. 
Implication gadgets whose $p$-vertices were identified with $r^j_{k_j}$ assert that $h_1(V^{j,k_j}_i) \in \mathcal{S}(H)$ and 
$h_1(U^{j,k_j}_i)$ is the signature of $h_1(V^{j,k_j}_i)$. Then, by \Cref{obs:sh-reverse}, $h_1(V^{j,k_j}_i)$ is a signature of $h_1(U^{j,k_j}_i)$.  
Denote $h_1(U^{j,k_j}_i)$ by $T^j_i$, then $h_1(V^{j,k_j}_i)=S(T^j_i)$ and $S(S(T^j_i))=T^j_i$.
Let $y^j_i=f^{j'}_{k_j}(x_i)$ be the \emph{candidate assignment} for $x_i \in X_{j'}$ at index $j$, recall that $y^j_i=\lambda^{-1}(S(T^j_i))$.
Let $i \in [n]$ be fixed and let $j_1, j_2 \in [L], j_1 < j_2$ be such that $x_i \in X_{j'_1} \cap X_{j'_2}$.
Observe that in such case $T^{j_1}_i\supseteq T^{j_2}_i$.
Indeed, denote $k_1=k_{j_1}$,  $k_2=k_{j_2}$, then we have (1) $h_1(V^{j_1,k_1}_i)=S(T^{j_1}_i) \ \textrm{ and } \ h_1(U^{j_1,k_1}_i)=T^{j_1}_i$, and (2) $h_1(V^{j_2,k_2}_i)=S(T^{j_2}_i) \ \textrm{ and } \ h_1(U^{j_2,k_2}_i)=T^{j_2}_i$.
Recall that each vertex from $U^{j_2,k_2}_i$ is adjacent to each vertex from $V^{j_1, k_1}_i$. 
Since $h_1$ is a homomorphism, the same holds for their images: each vertex from $T^{j_2}_i$ is adjacent to each vertex from $S(T^{j_1}_i)$. 
Then  $S(T^{j_2}_i) \supseteq S(T^{j_1}_i)$, so $T^{j_1}_i=S(S(T^{j_1}_i))\supseteq S(S(T^{j_2}_i)) \supseteq T^{j_2}_i$. 
We say that the index $j_1\in \{0,\ldots, L-1\}$ is \emph{problematic} for $i$ if there is $j_2>j_1$ such that $x_i \in X_{j_1'}\cap X_{j_2'} $ and $T^{j_1}_i \neq T^{j_2}_i$. Since for each variable we have at most $|H_1|$ problematic indices, there are at most $|H_1| \cdot n$ problematic indices for all variables.
Since $L=m(|H_i|\cdot n +1)$, by pigeonhole principle we get that there exists a set $J \subseteq  \{0,\ldots, L-1\}$ of $m$ consecutive indices such that none of them is problematic for any $i$.
For every $i \in [n]$, we fix some $j \in J$ such that $x_i \in X_{j'}$ and set $\gamma(x_i)=y^j_i$ (observe that the choice of $j$ does not matter).
We claim that $\gamma$ is an assignment that satisfies every constraint from $\varphi$. Indeed, for any $j' \in [m]$ there exists $j\in J$ such that $j'=j \mod m$. For every $i\in X_{j'}$, we have $\gamma(x_i)=y^j_i=f_{k_j}^{j'}(x_i)$, so $\gamma$ satisfies the constraint $c_{j'}$.
\end{proof}
Finally, it remains to adapt the arguments of Lampis~\cite{DBLP:journals/siamdm/Lampis20} to establish the desired linear clique-width bound.
\begin{claim}
\label{claim:cw-bound}
$G_{\phi}$ can be constructed in time polynomial in $|\phi|$, and we have $\cw{G_{\phi}} \le
n + f(\epsilon, \nu)$ for some function $f$, where $\nu=|V(H)|$.
\end{claim}
\begin{proof}[Proof of Claim]
Observe that any $(S,w,w')$-gadget constructed as in \Cref{lem:non-projective-s-gadget} for $i=1$ has at most $|V(H_i)|^{|S|-1}\cdot |V(H)|\le \nu^{|S|}$ vertices. 
In particular, we can ensure that every implication gadget in $G_{\phi}$ has at most $\nu^{\bigoh(\nu^2)}$ vertices. Moreover,
we will assume that all the or-gadgets of $G_{\phi}$ are constructed as in \Cref{lem:non-projective_or_gadget} and the subgadgets for $S$, $S_{\text{left}}$ and $S_{\text{right}}$ contain at most $\nu^7$ vertices. Then for every $j\in \{0,\ldots, L-1\}$, $p_{j'}$-or-gadget $(F_{j},h'_{j},R_j)$ has at most $ (p_{j'}-1)\cdot \nu^7$ vertices. 
For fixed $H$ and $\epsilon > 0$  we have that $B= s(H_i) \leq 2^{|V(H_i)|}-2$ and $q$ is a constant that only depends
on $B$, $\epsilon$ (that is, on $|V(H_i)|$, $\epsilon$). Each constraint of the \qcsp{q}{B} instance has at most $B^q$ satisfying
assignments. In particular, the number of vertices in each or-gadget is upper-bounded by $B^q \cdot \nu^7$. Therefore, it is not hard to see that the whole construction can be performed
in polynomial time, if $H$ is fixed and $\epsilon$ is a constant.
For clique-width we use the following labels:
\begin{enumerate}
\item $n$ \emph{main} labels, representing the variables of $\phi$.
\item A single \emph{done} label. Its informal meaning is that a vertex that receives this label will not
be connected to anything else not yet introduced in the graph.
\item $B^q \cdot \nu^7$ \emph{constraint work} labels.
\item $qB^q \cdot \nu^{\bigoh(\nu^2)}$ \emph{variable-constraint incidence work} labels.
\end{enumerate}
To give a clique-width expression we will describe how to build the graph, following
essentially the steps given in the description of the construction by maintaining the following
invariant: before starting iteration $j$, all vertices of the set $W^j_i = \bigcup_{j'<j}\bigcup_{k\in[p_{j'}]} V^{j',k}_i$ have label $i$, and all other vertices have the \emph{done} label.
This invariant is vacuously satisfied before the first iteration, since the graph is empty.
Suppose that for some $j \in \{0,\ldots, L-1\}$ the invariant is true. We use the $B^q \cdot \nu^7$ constraint work labels to introduce the vertices of the $p_{j'}$-or-gadget $(F_{j},h'_{j},R_j)$, giving each vertex a
distinct label. We use join operations to construct the internal edges of the or-gadget.
Then, for each variable $x_i$ that appears in the current constraint we do the following:
we use $B^q \cdot \nu^{\bigoh(\nu^2)}$ of the variable-constraint incidence work labels to introduce for all $k\in[p_{j'}]$ the vertices of $V_i^{j,k}$ and $U_i^{j,k}$
as well as the implication gadgets connecting these to $r^j_k$ . Again we use a
distinct label for each vertex, but the number of vertices (including internal vertices of the
implication gadgets) is $B^q \cdot \nu^{\bigoh(\nu^2)}$, so we have sufficiently many labels to use distinct labels for
each of the $q$ variables of the constraint. We use join operations to add the edges inside all
implication gadgets. Then we use join operations to connect $U_i^{j,k}$ to all vertices $W^j_i$. 
This is possible, since the invariant states that all the vertices of $W^j_i$ have
the same label $i$. We then rename all the vertices of $U_i^{j,k}$ for all $k$ to the done label, and
do the same also for internal vertices of all implication gadgets. 
We proceed to the next
variable of the same constraint and handle it using its own $B^q \cdot \nu^{\bigoh(\nu^2)}$ labels. Once we have
handled all variables of the current constraint, we rename all vertices of each $V_i^{j,k}$
to label $i$ for all $k$. We then rename all vertices of the $p_{j'}$-or-gadget $(F_{j},h'_{j},R_j)$ gadget to the done label and increase $j$
by 1. It is not hard to see that we have maintained the invariant and constructed all edges
induced by the vertices introduced in steps up to $j$, so repeating this process constructs the
graph.
\end{proof}
Together the claims imply \Cref{thm:lower-bound-homoext} in the following way:
For an arbitrary instance of \qcsp{q}{B}, our construction produces an instance of \hcoloringext{H}, and the instances are equivalent by \Cref{claim:equiv-forward} and \Cref{claim:equiv-backward}.
If one could solve \hcoloringext{H} in \(\bigohs((s(H_i) - \varepsilon)^{\cw{G}})\) for some \(\varepsilon > 0\), one could use our construction to solve \qcsp{q}{B}, and by our choice of \(B\) and \Cref{claim:cw-bound} this procedure would have complexity \(\bigohs((B - \varepsilon)^{n + c})\) for some constant \(c\).
By our choice of \(q\) according to \Cref{thm:q-exists}, this contradicts the SETH.
\end{proof} 

\section{Summary and Concluding Remarks}
\label{sec:conclusion}

\smallskip
\noindent \textsf{\bfseries Extensions and Corollaries.} \quad
We observe that Corollary \ref{cor:algorithm-factors} can be combined with Theorem~\ref{thm:lower-bound-homo} to obtain the following statement, which summarizes our results.

\begin{theorem}\label{thm:main}
Let $H'$ be a fixed graph with the non-trivial connected core $H$.
Let $H_1\times\ldots\times H_m$ be the factorization of $H$.
Let $i \in [m]$ be such that $s(H_i)=\max_{j \in [m]} s(H_j)$. 
Let $G$ be an instance of \homo{H'}.
\begin{enumerate}
\item Assuming a clique-width expression $\sigma$ of $G$ of width $\cw{G}$ is given, the \homo{H'} problem can be solved in time $\max_{i \in [m]}\Ohs(s(H_i)^{\cw{G}})$.
\item Assuming SETH, if $H$ is $H_i$-projective, then there is no algorithm to solve \homo{H'} in time $\max_{i \in [m]}\Ohs((s(H_i)-\epsilon)^{\cw{G}})$ for any $\epsilon >0$.
\end{enumerate}
\end{theorem}

We note that the restriction to connected targets can be avoided by known properties of  homomorphisms to disconnected graphs~\cite{OkrasaSODAJ}; on the algorithmic side, one branches over all connected components of $H$, while for the lower bound one considers the component with maximum signature number.

It is clear that obtaining a full complexity classification with respect to clique-width may require weakening the assumption in the second statement of \Cref{thm:main}.
We recall that an analogous situation occurs in the work of Okrasa and Rzążewski~\cite{OkrasaSODAJ}; as mentioned in the introduction, the authors obtain the SETH-conditioned tight complexity bound for the \homo{H} problem parameterized by treewidth for all targets $H$, assuming two conjectures of Larose and Tardif~\cite{LaroseT01,Larose02}.
The notion of $H_i$-projectivity allows us to restate these conjectures as one, which is not only sufficient in our setting but is also weaker in the sense of it being implied by the former two conjectures, but not necessarily equivalent to them.

\begin{myconjecture}\label{con:prime-projective}
Let $H$ be a non-trivial core with prime factorization $H_1 \times \ldots \times H_m$ and let $i \in [m]$. 
Then $H$ is $H_i$-projective.
\end{myconjecture}

Using \Cref{con:prime-projective}, we can restate our main result as follows. 

\begin{theorem}\label{thm:main-with-conjecture}
Let $H'$ be a fixed graph with the non-trivial connected core $H$.
Let $H_1\times\ldots\times H_m$ be the prime factorization of $H$.
Let $G$ be an instance of \homo{H'}.
\begin{enumerate}
\item Assuming the clique-width expression $\sigma$ of $G$ of width $\cw{G}$ is given, the \homo{H'} problem can be solved in time $\max_{i \in [m]}\Ohs(s(H_i)^{\cw{G}})$.
\item Assuming that \Cref{con:prime-projective} and SETH hold, there is no algorithm to solve \homo{H'} in time $\max_{i \in [m]}\Ohs((s(H_i)-\epsilon)^{\cw{G}})$ for any $\epsilon >0$.
\end{enumerate}
\end{theorem}

We also observe that since each non-trivial projective core $H$ is $H$-projective, in this case we already obtain a tight complexity bound.

\begin{corollary}
Let $H'$ be a fixed graph with the non-trivial connected projective core $H$.
Let $G$ be an instance of \homo{H'}.
\begin{enumerate}
\item Assuming the clique-width expression $\sigma$ of $G$ of width $\cw{G}$ is given, the \homo{H'} problem can be solved in time $\Ohs(s(H)^{\cw{G}})$.
\item There is no algorithm to solve \homo{H'} in time $\Ohs((s(H)-\epsilon)^{\cw{G}})$ for any $\epsilon >0$, unless the SETH fails.
\end{enumerate}
\end{corollary}

\smallskip
\noindent \textsf{\bfseries Generalizations and Other Research Directions.} \quad
We remark that our hardness reduction is via \hcoloringext{H}, and in fact our algorithm can also easily be adapted to this setting (by removing all records that do not adhere to the partial mapping from the input graph to \(H\)) without an increase in complexity.
However, since the dichotomy between \textsf{P} and \textsf{NP}-complete cases of \hcoloringext{H} is more complicated
 (see~\cite{FederV98}, studied as the graph-retract problem) 
 there exist target graphs $H$ that are not covered by \Cref{thm:main-with-conjecture}.
On a similar note, let us also point out that setting up the SETH-conditioned tight complexity bounds for clique-width for a more general list problem \lhomo{H}~\cite{FederHH03,PiecykR21} is widely open. 

Another direction that is very closely related to our results is to determine similarly tight complexity bounds for the \emph{rank-width} (\(\operatorname{rw}\)) of the input graph:
rank-width~\cite{OumS06,Oum05} is a graph parameter that is known to be asymptotically equivalent to clique-width and is in fact used as an approximation of clique-width that can be computed in fixed-parameter tractable time.
Our results together with the known relationship between clique-width and rank-width imply an upper bound of \(\mathcal{O}^*(s(H)^{2^{\operatorname{rw} + 1}})\) and a SETH lower bound of \((s(H) - \varepsilon)^{\operatorname{rw}}\) on the complexity of \homo{H} for projective \(H\) parameterized by the rank-width of the input.

\bibliographystyle{plainurl}
\bibliography{ref}
\end{document}